\documentclass[runningheads]{llncs}
\usepackage[dvipsnames]{xcolor}
\usepackage{enumitem}
\usepackage{graphicx,latexsym,amsfonts,amssymb,amsmath,wrapfig,subfigure}
\graphicspath{{./figures/}}

\begin{document}
\title{Simplifying Non-Simple Fan-Planar Drawings\thanks{This work was initiated at the
    $5^{th}$ DACH Workshop on Arrangements and Drawings, which was conducted online, via gathertown, in March 2021. The authors thank
    the organizers of the workshop for inviting us and providing a
    productive working atmosphere.
    B.\ K.\ was partially supported by DFG project WO$\,$758/11-1.
    M.\ M.\ R.\ is supported by the Swiss National Science Foundation within the collaborative DACH project \emph{Arrangements and Drawings} as SNSF Project 200021E-171681.
    K.~ K.\ is supported by DFG Project MU~3501/3-1 and within the Research Training Group GRK~2434 \emph{Facets of Complexity}. F.~S.\ is supported by the
German Research Foundation DFG Project FE 340/12-1.}}

%
%
\author{Boris Klemz\inst{1}\orcidID{0000-0002-4532-3765} \and
Kristin Knorr\inst{2}\orcidID{0000-0003-4239-424X} \and
Meghana M.\ Reddy\inst{3}\orcidID{0000-0001-9185-1246} \and 
Felix Schr\"oder\inst{4}\orcidID{0000-0001-8563-3517} }
\authorrunning{B. Klemz et al.}
%
\institute{Institut f\"ur Informatik, Universit\"at W\"urzburg, W\"urzburg, Germany \email{boris.klemz@uni-wuerzburg.de}\and
Institut f\"ur Informatik, Freie Universit\"at Berlin, Berlin, Germany \email{knorrkri@inf.fu-berlin.de}\and
Department of Computer Science, ETH, Z\"urich, Switzerland  \email{meghana.mreddy@inf.ethz.ch} \and
Institut f\"ur Mathematik, Technische Universit\"at Berlin, Berlin, Germany \email{fschroed@math.tu-berlin.de}}
\maketitle              

\begin{abstract} 
A drawing of a graph is fan-planar if the edges intersecting a common edge $a$ share a vertex $A$ on the same side of $a$. More precisely, orienting $e$ arbitrarily and the other edges towards $A$ results in a consistent orientation of the crossings. So far, fan-planar drawings have only been considered in the context of simple drawings, where any two edges share at most one point, including endpoints. We show that every non-simple fan-planar drawing can be redrawn as a simple fan-planar drawing of the same graph while not introducing additional crossings. Combined with previous results on fan-planar drawings, this yields that $n$-vertex-graphs having such a drawing can have at most $6.5n$ edges and that the recognition of such graphs is NP-hard.
We thereby answer an open problem posed by Kaufmann and Ueckerdt in 2014.

\keywords{Simple topological graphs  \and Fan-planar graphs \and Beyond-planar graphs \and Graph drawing.}
\end{abstract}
\section{Introduction}
\label{sec:intro}

In a \emph{fan-planar} drawing of a graph, each edge~$a$ is either not involved in any crossing or its crossing edges 
$c_1,\dots,c_k$ have a common endpoint~$A$ that is on a \emph{common side} of~$a$, 
i.e., orienting~$a$ arbitrarily and the edges $c_1,\dots,c_k$ towards~$A$ results in a consistent orientation of the crossings on~$a$ (either~$a$ crosses each $c_i$ from left to right at each crossing, or it crosses each $c_i$ from right to left at each crossing);
 for illustrations refer to Figure~\ref{fig:intro-examples}.
We call~$A$ the \emph{special} vertex of~$a$.
All \emph{graphs} in this paper are simple, that is, we do not allow parallel edges or self-loops. 
Hence, the vertex~$A$ is uniquely defined if~$k\ge 2$.
If $k=1$, then~$A$ is an arbitrary endpoint of~$c_1$. 

\begin{figure}[ht]
  \centering
  \includegraphics[scale=1,page=1]
  {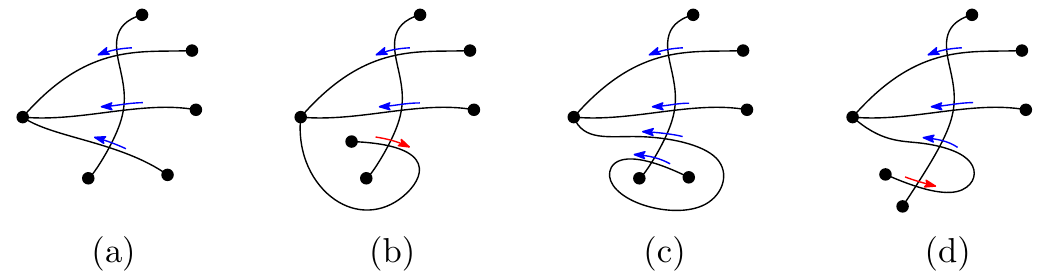}
  \caption{Drawings that are (a) simple and fan-planar, (b) simple and not fan-planar, (c) non-simple and fan-planar, and (d) non-simple and not fan-planar.  
  }
  \label{fig:intro-examples}
\end{figure}

Previous literature is exclusively concerned with fan-planar drawings that are also \emph{simple},
meaning that each pair of edges intersects in at most one point, which can be either an endpoint or a proper crossing.
Simple drawings can be characterized in terms of two forbidden crossing configurations\footnote{In the literature, usually more obstructions are mentioned, which we exclude for \emph{all} drawings (simple or not), see Section \ref{sec:terminology}. } (see Figure~\ref{fig:intro-forbidden}):
\begin{enumerate}[leftmargin=*,label={S\arabic*}]
\item Two adjacent edges cross.
\item Two edges cross at least twice.
\end{enumerate}
Simple drawings that are fan-planar can be characterized in terms of two additional forbidden crossing configurations~\cite{KU_dfpg_2014} (see Figure~\ref{fig:intro-forbidden}):
\begin{enumerate}[leftmargin=*,label={SF\arabic*}]
\item Two independent edges cross a common third edge.
\item Two adjacent edges cross a third edge $a$ such that their common endpoint~$A$ is not on a common side of~$a$.
\end{enumerate}
In this paper, we study non-simple fan-planar drawings and how to turn them into simple fan-planar drawings.

\begin{figure}[ht]
  \centering
  \includegraphics[scale=1,page=6]
  {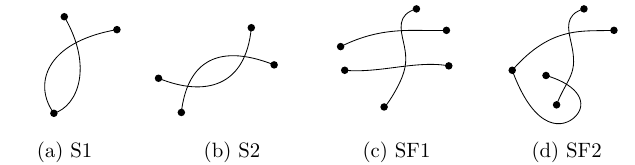}
  \caption{Forbidden configurations in simple fan-planar drawings.
  }
  \label{fig:intro-forbidden}
\end{figure}
\paragraph{Previous and related work.}
A drawing is $k$-\emph{planar} if each edge is crossed at most $k$ times and a graph is $k$-\emph{planar}, if it admits such a drawing~\cite{DBLP:journals/combinatorica/PachT97}.
A $k$-\emph{quasiplanar} graph can be drawn such that no $k$ edges mutually cross -- such a drawing is called $k$-\emph{quasiplanar}~\cite{DBLP:journals/combinatorica/AgarwalAPPS97}.
Kaufmann and Ueckerdt~\cite{KU_dfpg_2014} introduced the notion of fan-planarity in $2014$.
They describe the class of graphs representable by simple fan-planar drawings\footnote{In~\cite{KU_dfpg_2014}, these graphs are called \emph{fan-planar}. We do not use this terminology to avoid mix-ups with the class of graphs admitting (not necessarily simple) fan-planar drawings.} as somewhere between $1$-planar graphs and $3$-quasiplanar graphs.
Indeed, every $1$-planar graph admits a simple $1$-planar drawing.
Since such a drawing cannot contain configuration~SF1 or~SF2, it is fan-planar.
Moreover, a simple fan-planar drawing cannot contain three mutually crossing edges and, therefore, it is $3$-quasiplanar.
Binucci et al.~\cite{BGDMPST_fppc_2015} have shown that for each $k\ge 2$ the class of graphs admitting simple $k$-planar drawings and the class of graphs admitting simple fan-planar drawings are incomparable.
In contrast, every so-called optimal $2$-planar graph admits a simple fan-planar drawing~\cite{DBLP:conf/compgeom/Bekos0R17}. This follows from the fact that these graphs can be characterized as the graphs obtained by drawing a pentagram in the interior of each face of a pentagulation~\cite{DBLP:conf/compgeom/Bekos0R17}, which yields a fan-planar drawing.
Angelini et al.~\cite{DBLP:journals/tcs/AngeliniBKKS18} introduced a drawing style that combines fan-planarity with a visualization technique called edge bundling~\cite{DBLP:journals/tvcg/Holten06,DBLP:journals/cgf/HoltenW09,DBLP:journals/cgf/TeleaE10}.
Each of their so-called $1$-sided $1$-fan-bundle-planar drawings represents a graph that is also realizable as a simple fan-planar drawing, but the converse is not true~\cite{DBLP:journals/tcs/AngeliniBKKS18}. Brandenburg \cite{Bran2020} examines \emph{fan-crossing} drawings, where all edges crossing a common edge share a common endpoint (in particular, this implies that SF1 is forbidden), as well as \emph{adjacency-crossing} drawings, where SF1 is the only obstruction.
Simple fan-planar drawings are somewhat opposite to simple $k$-\emph{fan-crossing-free}~\cite{DBLP:journals/algorithmica/CheongHKK15} drawings, where no $k\ge 2$ adjacent edges cross another common edge.

The maximum number of edges in a simple fan-planar drawing on $n$ vertices is upperbounded by $6.5n-20$~\cite{KU_dfpg_2014}, which follows from the known density bounds for $3$-quasiplanar graphs~\cite{DBLP:journals/jct/AckermanT07}.
A better upper bound of $5n-10$ edges was claimed in a preprint~\cite{KU_dfpg_2014}.
However, the corresponding proof appears to be flawed.
We spoke with the authors and they confirmed that the current version of their proof is not correct and that they do not see a simple way to fix it\footnote{More specifically, the statement and proof of~\cite[Lemma 1]{KU_dfpg_2014} are incorrect. A counterexample can be obtained by removing the edge~$g$ from the construction illustrated in Figure~\ref{fig:reroute_b_fig} (vertices $R,B$ correspond to the vertices $u,w$ in~\cite[Lemma 1]{KU_dfpg_2014}); for a formal description of the construction see Lemma~\ref{lem:sequence_of_edges}.

After our submission to GD'21, the authors of~\cite{KU_dfpg_2014} have uploaded a new version~\cite{KU_dfpg_2014V2} of their preprint in which they state a different definition of fan-planarity with an additional forbidden crossing configuration; also see~\cite[last paragraph of Section~1]{KU_dfpg_2014V2}.}.
Kaufmann and Ueckerdt~\cite{KU_dfpg_2014} described an infinite family of simple fan-planar drawings with $5n-10$ edges.
The same lower bound also follows from the aforementioned connection to optimal $2$-planar graphs~\cite{DBLP:conf/compgeom/Bekos0R17}.

The recognition of graphs realizable as simple fan-planar drawings is NP-hard~\cite{BGDMPST_fppc_2015}.
The same statement also holds in the fixed rotation system setting~\cite{DBLP:journals/algorithmica/BekosCGHK17}, where the cyclic order of edges incident to each vertex is prescribed as part of the input.
Consequently, efficient algorithms have only been discovered for special graph classes~\cite{DBLP:journals/algorithmica/BekosCGHK17} and for restricted drawing styles
\cite{DBLP:journals/algorithmica/BekosCGHK17,DBLP:journals/jgaa/BinucciCDGKKMT17}.

For a more comprehensive overview of previous work related to fan-planarity, we refer to a very recent survey article dedicated to fan-planarity due to Bekos and Grilli~\cite{DBLP:books/sp/20/Bekos020}.
The study of fan-planarity also falls in line with the recent trend of studying so-called beyond-planar graph classes, whose corresponding drawing styles permit crossings in restricted ways only.
Apart from $k$-planar~\cite{DBLP:journals/combinatorica/PachT97}, $k$-quasiplanar~\cite{DBLP:journals/combinatorica/AgarwalAPPS97}, $k$-fan-crossing-free~\cite{DBLP:journals/algorithmica/CheongHKK15}, fan-bundle-planar~\cite{DBLP:journals/tcs/AngeliniBKKS18}, fan-crossing~\cite{Bran2020}, adjacency-crossing~\cite{KU_dfpg_2014}, and fan-planar~\cite{KU_dfpg_2014} drawings, which have already been mentioned above, several other classes of beyond-planar graphs and their corresponding drawing styles have been studied, e.g.: $k$-gap-planar drawings~\cite{DBLP:journals/tcs/BaeBCEE0HKMRT18} (each crossing is assigned to one of the involved edges such that each edge is assigned at most~$k$ crossings), RAC-drawings~\cite{DBLP:journals/tcs/DidimoEL11} (straight-line drawings with right angle crossings), and many more.
We refer to~\cite{DBLP:journals/csur/DidimoLM19,DBLP:books/sp/20/HT2020} for recent surveys on beyond-planar graphs.

\paragraph{Contribution.}
A fan-planar drawing that is not simple may contain configuration~S1.
Configuration~S2 is allowed in a partial sense:
two edges may cross any number of times, but only if orienting them arbitrarily results in a consistent orientation of their crossings; cf. Figures~\ref{fig:intro-examples}(c) and~(d).
Recall that every simple fan-planar drawing is $3$-quasiplanar.
In contrast, Figure~\ref{fig:non_simple_fan_planar}(a) depicts a non-simple fan-planar drawing that is not $3$-quasiplanar, which suggests that graphs admitting non-simple fan-planar drawings are not necessarily $3$-quasiplanar.
Consequently, the density bound of $6.5n-20$~\cite{DBLP:journals/jct/AckermanT07} for 3-quasiplanar graphs does not directly carry over.
However, the depicted graph is just a~$K_3$, which can obviously be redrawn as a simple (fan-)planar drawing.
This raises two very natural questions:

\begin{enumerate}
\item Is the largest number of edges in a $n$-vertex non-simple fan-planar drawing larger than the number of edges in any $n$-vertex simple fan-planar drawing?
\item Which non-simple fan-planar drawings can be redrawn as simple fan-planar drawings of the same graph?
\end{enumerate}

\begin{figure}[ht]
  \centering
  \includegraphics[scale=1,page=7]
  {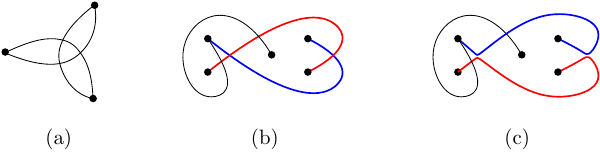}
  \caption{(a) A $3$-quasiplanar non-simple fan-planar drawing. (b) A non-simple fan-planar drawing. Applying the standard procedure for simplifying configuration~S2 yields the drawing in (c), which is not fan-planar since the black edge crosses two independent edges.}
  \label{fig:non_simple_fan_planar}
\end{figure}

Question~1 is also mentioned as an open problem by Kaufmann and Ueckerdt~\cite{KU_dfpg_2014}.
Regarding question~2, we remark that the standard method for simplifying the configurations~S1 and S2 does not necessarily maintain fan-planarity, see Figures~\ref{fig:non_simple_fan_planar}(b) and~(c).
As our main result, we answer both questions, thereby solving the open problem by Kaufmann and Ueckerdt:

\begin{theorem} \label{thm:main}
Every non-simple fan-planar drawing can be redrawn as a simple fan-planar drawing of the same graph without introducing additional crossings.
\end{theorem}

The proof of Theorem~\ref{thm:main} is constructive and gives rise to an efficient algorithm for simplifying a given fan-planar drawing. Combined with the aforementioned previous results regarding the density~\cite{DBLP:journals/jct/AckermanT07,KU_dfpg_2014} and the recognition complexity~\cite{BGDMPST_fppc_2015} of graphs realizable as simple fan-planar drawings, we obtain:

\begin{corollary}
Every (not necessarily simple) fan-planar drawing realizes a $3$-quasiplanar graph.
\end{corollary}

\begin{corollary}
Every (not necessarily simple) fan-planar drawing on $n$ vertices has at most $6.5n-20$ edges.
\end{corollary}

\begin{corollary}
Recognizing graphs that admit (not necessarily simple) fan-planar drawings is NP-hard.
\end{corollary}

We start with some basic terminology and conventions in Section~\ref{sec:terminology}.
The algorithm for simplifying non-simple fan-planar drawings is described in Section~\ref{sec:main}.

\section{Terminology}
\label{sec:terminology}

In all drawings in this paper, edges are represented by simple curves.
We assume no two edges touch, that is, meet tangentially.
Further, we assume that no three edges share a common crossing and that edges do not contain vertices except their endpoints.
Let~$\Gamma$ be a drawing of a graph~$G$.
A \emph{redrawing} of~$\Gamma$ is a drawing of~$G$.
\emph{Redrawing} an edge~$e$ in~$\Gamma$ refers to the process of obtaining a redrawing~$\Gamma'$ of~$\Gamma$ such that $(\Gamma -e)=(\Gamma'-e)$.

In the beginning of Section~\ref{sec:intro}, we introduced the notion of special vertices for crossed edges.
To streamline the arguments, we also assign an arbitrarily chosen \emph{special} vertex to each uncrossed edge.
Let~$e$ and~$f$ be edges that cross and let~$E$ be the special vertex of~$e$.
We define the $i^{th}$ crossing of $f$ with $e$ as the $i^{th}$ crossing between $f$ and $e$ encountered when traversing $f$ from endpoint~$E$.
For example, in Figure~\ref{fig:multi_crossings_fig1}, the first crossing of $g$ with $b$ is $x$ and the second crossing is $y$.

\section{The Redrawing Procedure}
\label{sec:main}

We prove Theorem~\ref{thm:main} by providing an algorithm that redraws the edges of a non-simple fan-planar drawing~$\Gamma$ to obtain a simple fan-planar drawing.
 It is based on three subroutines (Lemmata~\ref{lem:adjacent_crossings_case1}, \ref{lem:multi_crossings} and~\ref{lem:adjacent_crossings_case2}), which can be iteratively applied to remove crossings between adjacent edges (configuration S1)
and
multiple crossings between pairs of edges (configuration S2).
More specifically, the first procedure (Lemma~\ref{lem:adjacent_crossings_case1}) eliminates a particular type of adjacent crossings, namely, those that involve an edge that is incident to its special vertex.
The second procedure (Lemma~\ref{lem:multi_crossings}) removes multiple crossings between edge pairs.
Both procedures reduce the overall number of crossings.
Hence, they can be exhaustively applied to obtain a redrawing~$\Gamma'$ of~$\Gamma$ that does not contain multiple crossings between edge pairs and where adjacent crossings only involve edges that are not incident to their special vertices (Corollary~\ref{cor:cor1}).
The procedure (Lemma~\ref{lem:adjacent_crossings_case2}) for removing these remaining crossings is quite involved and based on a structural analysis (Lemma~\ref{lem:sequence_of_edges}) of the drawing~$\Gamma'$.

The first procedure, for getting rid of some of the adjacent crossings, is very simple; the proof is deferred to Appendix~\ref{app:lem1}, but illustrated in Figure~\ref{fig:adjacent_crossings_case1body}.

\begin{lemma} \label{lem:adjacent_crossings_case1}
Let $\Gamma$ be a non-simple fan-planar drawing. Let $b=(B,R)$ be an edge in $\Gamma$ that is incident to its special vertex $B$. If $b$ has at least one crossing, then one of the edges in the drawing can be redrawn such that the total number of crossings in the drawing decreases.
Moreover, the redrawing is fan-planar.
\end{lemma}

\begin{figure}[ht]
  \centering
  \subfigure[]{
	 \label{fig:adjacent_crossings_case1_fig1body}
	\includegraphics[scale=1]
	{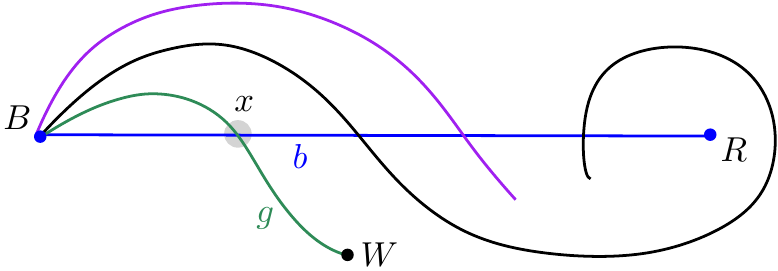}}

  \subfigure[]{
	 \label{fig:adjacent_crossings_case1_fig2body}
	\includegraphics[scale=1]
	{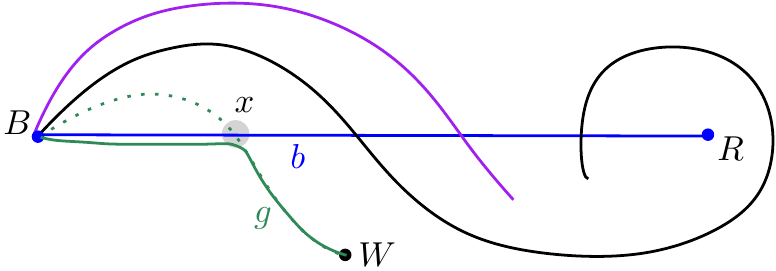}}

    \caption{Illustration of Lemma \ref{lem:adjacent_crossings_case1}. If $b$ is incident to its special vertex $B$, then all crossings on $b$ are adjacent crossings. We redraw the edge $g$ whose crossing $x$ with $b$ is closest to $B$ along $b$. Redrawing the part of $g$ between $x$ and $B$ along~$b$ cannot introduce any new crossings.}
  \label{fig:adjacent_crossings_case1body}
\end{figure}

We continue by describing the second procedure, which eliminates crossings between pairs of edges (independent or adjacent) that cross more than once.

\begin{lemma} \label{lem:multi_crossings} 
Let $\Gamma$ be a non-simple fan-planar drawing. Let $b=(G,R)$ be an edge in $\Gamma$ whose special vertex $B$ is not incident to $b$. 
If edge~$b$ has multiple crossings with at least one other edge, then an edge that crosses~$b$ multiple times, say $g=(B,W)$ (where $W$ could also be incident to $b$), can be redrawn such that at least one crossing between $b$ and $g$ is eliminated and the total number of crossings in the drawing decreases. 
Moreover, the redrawing is fan-planar.
\end{lemma}

\begin{figure}[ht]
  \centering
  \subfigure[]{
	 \label{fig:multi_crossings_fig1}
	\includegraphics[scale=1]
	{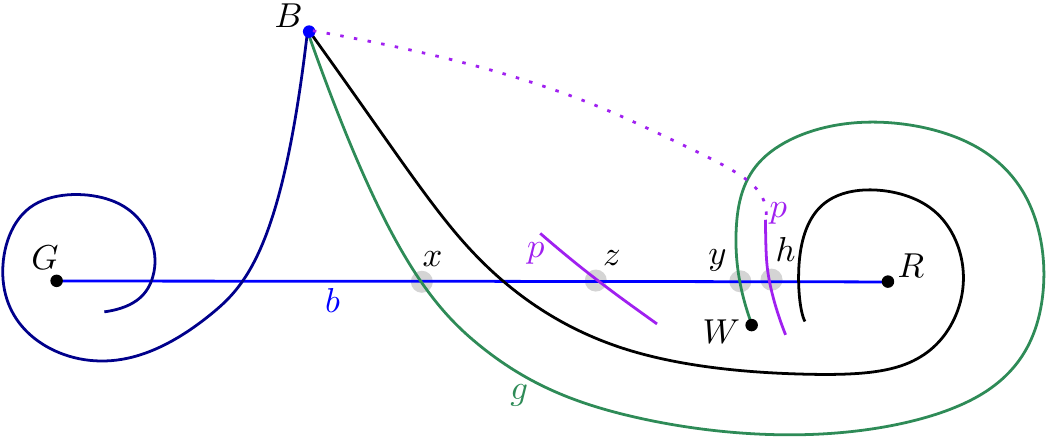}}

  \subfigure[Redrawing of edge $g$.]{
	 \label{fig:multi_crossings_fig2}
	\includegraphics[scale=1]
	{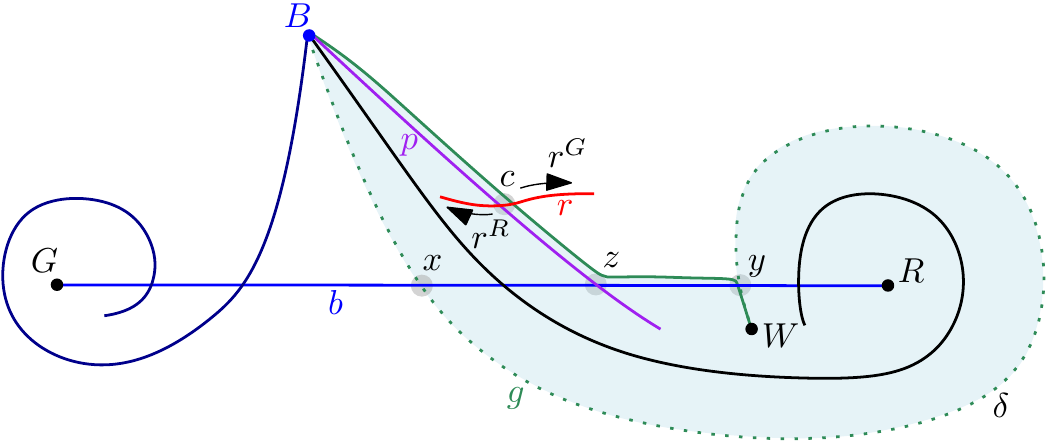}}

    \caption{Edges $b$ and $g$ cross multiple times and the special vertex $B$ of $b$ is not incident to $b$. The figures do not reflect the case when $W$ is incident to~$b$. Any new crossing with the redrawn version of $g$ involves an edge $r$ crossing $p$ between $z$ and $B$, which has to cross the replaced part of $g$ since it is incident to $G$ or $R$.}
  \label{fig:multi_crossings}
\end{figure}

\begin{proof}
We start by describing a procedure to pick the edge that will be redrawn.
We traverse $b$ from $G$ to $R$, until the second crossing of an edge $g = (B, W)$ with~$b$ is encountered such that the first crossing of $g$ with~$b$ appeared before its second crossing, i.e., the second crossing~$y$ with $b$ is closer to $R$ than its first crossing~$x$ with~$b$, see Figure~\ref{fig:multi_crossings_fig1}.
If no such edge exists, we exchange the roles of $R$ and $G$ and repeat the procedure.
We are guaranteed to find an edge~$g$ with the desired properties, since there is an edge crossing~$b$ multiple times.

So without loss of generality, assume that the edge $g$ has its second crossing~$y$ with $b$ closer to $R$ than its first crossing~$x$.
We then walk from $y$ towards $G$ along $b$ until we encounter a crossing $z$ between an edge $p$ and $b$.
The edge $p$ must also be incident to $B$, the special vertex of~$b$.

We can now describe the redrawing procedure.
The edge $g$ is redrawn to follow its previous drawing from $W$ to $y$, cross $b$ at $y$, follow the drawing of $b$ from $y$ to $z$, and, finally, closely follow $p$ from $z$ to $B$; for an illustration see Figure~\ref{fig:multi_crossings}.
The following statements are proved in Appendix~\ref{app:prop12}.

\begin{proposition}
The crossing $z$ is the first crossing of $p$ with~$b$.
\label{prop:z}
\end{proposition}

\begin{proposition}
Redrawing $g$ maintains fan-planarity. Moreover, there is an injective mapping that assigns each crossing on the redrawn part of $g$ to a crossing on the replaced part of $g$ that involves the same edges.
\label{prop:multifan}
\end{proposition}

The described redrawing of $g$ eliminates the crossing between $g$ and $b$ at $x$.
Moreover, by Proposition~\ref{prop:z}, the redrawn version of~$g$ does not contain any new crossings between $b$ and $g$.
Combined with Proposition~\ref{prop:multifan}, it follows that the total number of crossings decreases.
Moreover, fan-planarity is maintained.
\hfill$\square$
\end{proof}

Equipped with Lemmata~\ref{lem:adjacent_crossings_case1} and \ref{lem:multi_crossings}, we can apply the following normalization to the drawing (the proof can be found in Appendix~\ref{app:cor1}):

\begin{corollary} \label{cor:cor1}
Let $\Gamma$ be a non-simple fan-planar drawing. There is a fan-planar redrawing $\Gamma'$ of $\Gamma$ such that
\begin{itemize}
\item no two edges cross more than once in $\Gamma'$;
\item no edge is incident to its special vertex; and
\item $\Gamma'$ does not have more crossings than $\Gamma$.
\end{itemize}
\end{corollary}

Adjacent crossings between edges that are not incident to their special vertices may lead to configurations where the previous edge-rerouting strategies would incur additional crossings.
In the following lemma, we deal with some unproblematic cases and characterize the remaining, more challenging, configurations in terms of a sequence of conflicting edges.

\begin{lemma} \label{lem:sequence_of_edges}
Let $\Gamma$ be a non-simple fan-planar drawing in which no two edges cross more than once and such that no edge is incident to its special vertex.
Let $b=(G,R)$ and $g = (R, B)$ be (adjacent) edges which cross each other at $x$.

We can redraw~$g$ such that the total number of crossings decreases and fan-planarity is maintained; or, alternatively, we can determine a sequence of edges $r_0,b_1, r_2, b_3, r_4, \dots, r_k$  such that the edges $b,g,r_0,b_1, r_2, b_3, r_4, \dots, r_k$ are pairwise distinct and the following properties are satisfied (we call the edges $r_i$``red'' and the edges $b_i$``black''; for an illustration, see Figure~\ref{fig:adjacent_crossings_case2_descr}, as well as Figure~\ref{fig:reroute_b_fig}, which also depicts $r_k$):

\begin{enumerate}[leftmargin=*,label={I\arabic*}]
\item\label{i:B} $B$ is the special vertex of the black edges and incident to the red edges.
\item\label{i:R} $R$ is the special vertex of the red edges and incident to the black edges.
\item\label{i:un} For any odd $i$, the first crossing $x_{i+1}$ of $b_i$ starting from $R$ is with $r_{i+1}$.
 For any even $i<k$, the first crossing $x_{i+1}$ of $r_i$ starting from $B$ is with $b_{i+1}$.
\item\label{i:exit} $r_0$ crosses $b_1$ but no other black edge. $b$ crosses $r_0$ and $r_k$ but no other red edges. 

\item\label{i:f} For the purposes of the final two invariants, we define $q_{-1}=b$. For $0\leq i < k$, let $\alpha_i$ be the closed curve defined by $g$, the arc of $q_i$ and the arc of $q_{i-1}$, where $q \in \{r, b\}$, that connect $R,B$ and $x_i$.
For $0\le i<k$, let $\Gamma_i$ be the drawing induced by the edges $b$, $g, r_0,b_1, r_2, \dots, q_i$.

For $0\le i<k$, the curve $\alpha_i$ is simple and bounds a region $f_i$ that contains only $G$, an arc of~$b$ that connects~$G$ to~$x\in\alpha_i$ and, possibly, an arc of~$r_0$ that connects~$G$ to~$\alpha_i$, in its interior.
\item \label{i:tri} For $0<i < k$, $f_{i} \subset f_{i-1}$ and $f_{i-1} \setminus f_{i}$ is an empty triangular face in $\Gamma_{i}$ bounded by the following three arcs:
	\begin{itemize}
		\item the arc of $q_{i}$ between $x_{i}$ and the special vertex of $q_{i-1}$,
		\item the arc of $q_{i-1}$ between $x_{i}$ and $x_{i-1}$,
		\item the arc of $q_{i-2}$ between $x_{i-1}$ and the special vertex of $q_{i-1}$
	\end{itemize}
where $q \in \{r, b\}$.
\end{enumerate}
\end{lemma}

\begin{remark} \label{rem:g_crosses_b}
Note that invariant \ref{i:f} implies that in $\Gamma_i$, $g$ crosses only $b$ and possibly~$r_0$.
Moreover, the arcs of $q_i$ and $q_{i-1}$ connecting $R$ and $B$ via $x_i$ are uncrossed in $\Gamma_i$.
\label{rem:tri}
\end{remark}

\begin{figure}[ht]
  \centering
  \includegraphics[scale=1]
  {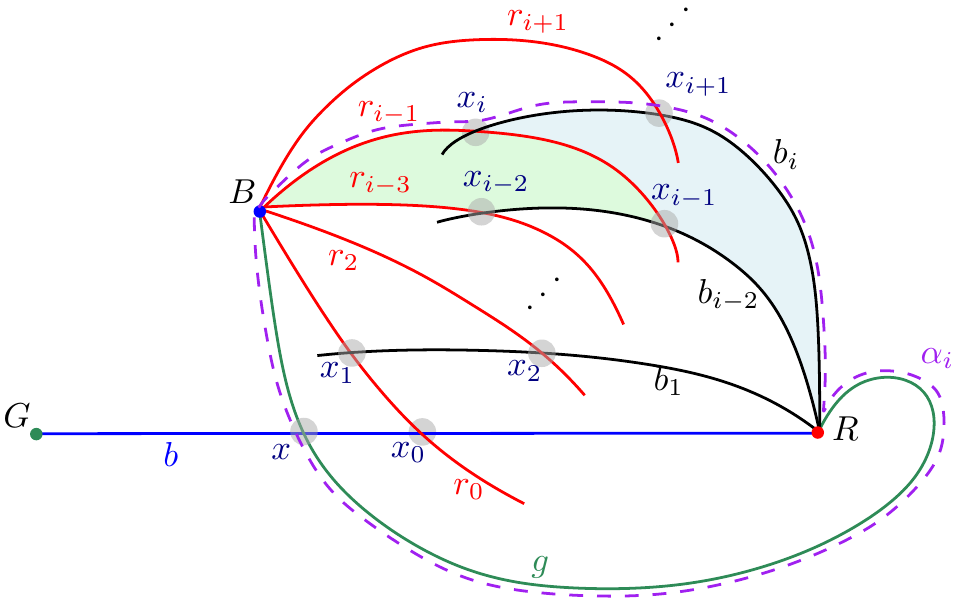}
  \caption{An example of the sequence of edges described in Lemma~\ref{lem:sequence_of_edges}. The face $f_i$ is the unbounded region delimited by the dashed curve, the face $f_{i-1} \setminus f_{i}$ is depicted in blue and the face $f_{i-2} \setminus f_{i-1}$ is depicted in green.}
  \label{fig:adjacent_crossings_case2_descr}
\end{figure}

\begin{proof}
It follows from the preconditions that~$B$ is the special vertex of~$b$ and~$G$ is the special vertex of~$g$. 
We will construct the sequence of edges inductively.

\paragraph{Base case.}
For the induction base case, we show how to determine~$r_0$ and~$b_1$ such that all invariants are satisfied with respect to~$r_0$.
For $b_1$, we will only establish the invariants~\ref{i:R}, \ref{i:exit}, \ref{i:f} and \ref{i:tri}.

We traverse from $R$ along $b$ until we encounter an edge $r_0$ that crosses $b$ and denote its crossing by $x_0$.
If~$x_0=x$ and, hence, $r_0=g$, we can redraw the part of~$g$ that leads from~$R$ to~$x$ along~$b$ such that the crossing at~$x$ is removed.
Moreover, since the redrawn part is crossing-free, the total number of crossings is decreased and fan-planarity is maintained.
Hence, if~$x_0=x$, the statement of the lemma holds.

So assume that~$x_0\neq x$.
It follows that, $r_0\neq g$ since~$g$ cannot cross $b$ multiple times.
Moreover, $r_0\neq b$ since edges are realized as simple curves.
Since~$r_0$ intersects~$b$, it is incident to~$B$.

Now, we traverse~$r_0$ from~$B$ towards~$x_0$ until we encounter a crossing~$x_1$ with an edge~$b_1$.
If~$x_1=x_0$ and, hence, $b_1=b$, we redraw~$g$ along the part of~$b$ between~$R$ and~$x_0$ and the part of~$r_0$ between~$x_0$ and~$B$.
The redrawn version of~$g$ is crossing-free.
Hence, we have eliminated at least one crossing (namely~$x$) while maintaining fan-planarity
and, thus, the statement of the lemma holds if~$x_1=x_0$.

So assume that~$x_1\neq x_0$.
It follows that $b_1$ is distinct from~$b$ since~$b$ has no multiple crossings with~$r_0$.
Moreover, $b_1\neq r_0$ since edges are simple curves.
Finally, we show that~$b_1\neq g$.
In fact, we actually claim something stronger and prove it in Appendix~\ref{app:base}.

\begin{proposition}
\label{prop:base}
The part of~$r_0$ between~$B$ and~$x_0$ cannot cross~$g$.
\end{proposition}

In particular, Proposition~\ref{prop:base} implies~$b_1\neq g$, as claimed. Thus, we have determined two edges~$r_0$ and~$b_1$ such that~$b,g,r_0,b_1$ are pairwise distinct.
It remains to show that the desired invariants hold.
We have already established that~$r_0$ is incident to~$B$ (since it intersects~$b$) and, thus, \ref{i:B} is satisfied for~$r_0$.

Since~$b_1$ and~$b$ cross~$r_0$, it follows that~$b_1$ shares an endpoint with~$b$, which is the special vertex of~$r_0$.
Accordingly, we consider two cases.
First, assume that the special vertex of $r_0$ is $G$, which is illustrated in Figure~\ref{fig:adjacent_crossings_base}.
Consider the closed curve~$\alpha_0$ described by~$g$, the part of~$r_0$ between~$x_0$ and~$B$ and the part of~$b$ between~$R$ and~$x_0$.
By Proposition~\ref{prop:base} and the fact that there are no multiple crossings, the curve~$\alpha_0$ is indeed simple.
Orient~$b$ and~$b_1$ towards~$G$.
Since the resulting orientation of the crossings~$x_0$ and~$x_1$ has to be consistent, it follows that the part of~$b_1$ that connects~$x_1$ with~$G$ has to intersect~$\alpha_0$.
More specifically, since there are no multiple crossings, it needs to intersect~$g$ in some point~$z$.
We now redraw~$g$ along the part of~$b$ that connects~$R$ with~$x_0$ and the part of~$r_0$ that connects~$x_0$ with~$B$.
The redrawn version of~$g$ only has crossings along the part between~$x_0$ and~$B$.
In particular, it crosses~$b_1$ at~$x_1$, but the orientation of this crossing is consistent with the orientation of~$z$ in the original drawing of~$g$.
The same argument applies for all other intersected edges.
Consequently, we introduce no additional crossings, eliminate the crossing~$x$, and maintain fan-planarity.
Hence, the statement of the lemma holds if the special vertex of~$r_0$ is~$G$.
It remains to consider the case where the special vertex of~$r_0$ is~$R$ and, hence, $b_1$ is incident to~$R$.
It follows that invariant~\ref{i:R} is satisfied for both~$r_0$ and~$b_1$.

Invariant~\ref{i:un} for~$r_0$ is satisfied by construction (and for~$b_1$ there is nothing to show). Invariant~\ref{i:exit} is also satisfied for~$r_0$ and~$b_1$ by construction.

The edge~$r_0$ cannot cross~$b$ or~$b_1$ a second time.
If it crosses~$g$, then it is incident to~$G$, the special vertex of~$g$.
In any case, this implies invariant~\ref{i:f} for~$\Gamma_0$.

We observe that~$b_1$ cannot cross~$b$ or~$g$ since this would imply that~$b_1$ is incident to~$B$ or~$G$ (the special vertex of~$b$ and~$g$, respectively) and hence $b_1$ is parallel to~$g$ or~$b$, respectively.
Moreover, $b_1$ cannot cross~$r_0$ a second time.
Hence, the part of~$b_1$ that leads from~$R$ to~$x_1$ is crossing-free in the drawing~$\Gamma_1$.
Together with invariant~\ref{i:f} for~$\Gamma_0$, the invariant~\ref{i:f} holds for~$\Gamma_1$ and invariant~\ref{i:tri} holds, which concludes the base case.

\begin{figure}[ht]
  \centering
  \includegraphics[scale=1]
  {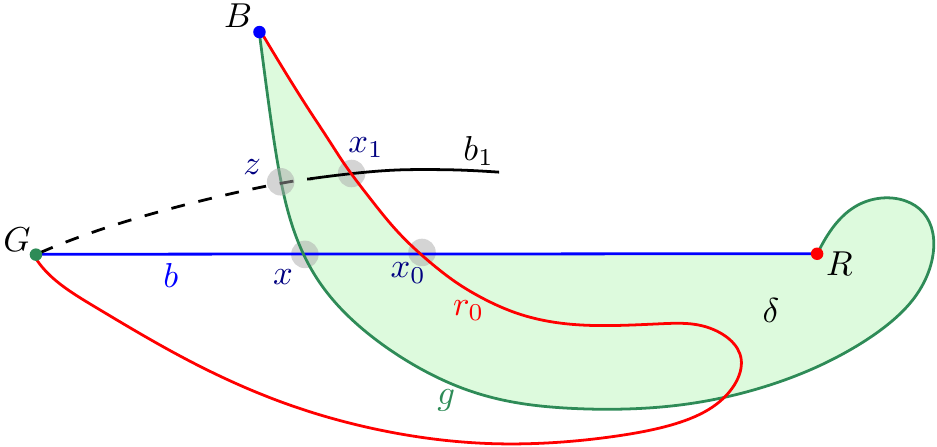}
  \caption{$r_0$ can be incident to $G$. $b_1$ is drawn as if $G$ is the special vertex of $r_0$.}
  \label{fig:adjacent_crossings_base}
\end{figure}

\paragraph{Inductive Step:} Now, assume the first $j+1$ edges, $r_0, b_1, \ldots, q_j$, have been determined and $j<k$.
We assume all invariants hold for $r_0, \dots q_{j-1}$.
Additionally, we assume that~\ref{i:B}, \ref{i:exit}, \ref{i:f} and \ref{i:tri} hold for $q_j$ if $j$ is even (and hence $q_j=r_j$ is red), 
or~\ref{i:R}, \ref{i:exit}, \ref{i:f} and \ref{i:tri} hold for $q_j$ if $j$ is odd (and hence $q_j=b_j$ is black).

We will now determine the edge $q_{j+1}$.
If $j$ is even, we need to prove the invariants~\ref{i:R} and~\ref{i:un} for $q_j$ and the invariants~\ref{i:R}, \ref{i:exit}, \ref{i:f} and \ref{i:tri} for $q_{j+1}$.
If $j$ is odd, we need to prove the invariants~\ref{i:B} and~\ref{i:un} for $q_j$ and the invariants~\ref{i:B}, \ref{i:exit}, \ref{i:f} and \ref{i:tri} (and~\ref{i:R} and~\ref{i:un} if $j+1=k$) for $q_{j+1}$.

\paragraph{Case 1:} $q_j = r_j$. Note that in this case, we have nothing to prove for invariant~\ref{i:B}.

If the edge $r_j$ has no crossings between $B$ and $x_j$, then we could redraw $g$ along this part of $r_j$ and the arc of $b_{j-1}$ from $x_j$ to $R$. The redrawn version of $g$ would then be uncrossed by invariant~\ref{i:un} for $b_{j-1}$ and the lemma is proved.

Otherwise $r_j$ has at least one crossing between $B$ and $x_j$.
 We determine edge $b_{j+1}$ as follows: traverse along $r_j$ from $B$ towards $x_j$ until we encounter the first edge that crosses $r_j$, let this edge be $b_{j+1}$. 

Invariant~\ref{i:un} for $r_j$ is satisfied by construction.
To prove the remaining invariants, we establish several propositions, the proofs of which can be found in Appendix~\ref{app:case1}. First we prove invariant~\ref{i:R} for $r_j$ and $b_{j+1}$.

\begin{proposition}
\label{prop:R}
Edge $b_{j+1}$ is incident to $R$.
\end{proposition}

Since $r_j$ crosses both edges $b_{j+1}$ and $b_{j-1}$, which are both incident to $R$, the special vertex of $r_{j}$ is $R$, which proves invariant~\ref{i:R}.

Next we prove invariant~\ref{i:exit}. We first need to prove $b_{j+1}$ is distinct from all the previous edges.

\begin{proposition} \label{prop:r0}
 $b_{j+1}$ is distinct from all edges of $\Gamma_j$ and does not cross edge $r_0$.
\end{proposition}

Lastly, to prove invariants~\ref{i:f} and~\ref{i:tri}, we have to prove the following proposition.

\begin{proposition}
The arc of $b_{j+1}$ between $R$ and $x_{j+1}$ is uncrossed in the drawing~$\Gamma_{j+1}$. \label{prop:uncross}
\end{proposition}

So the arc of $b_{j+1}$ between $x_{j+1}$ and $R$ is uncrossed in the drawing $\Gamma_{j+1}$. 
Further, the arc of $r_j$ between $x_j$ and $B$ was uncrossed in $\Gamma_j$ as noted in Remark~\ref{rem:tri}. Since $b_{j+1}$ is the only new edge introduced for $\Gamma_{j+1}$, the arcs of $r_j$ between $x_j$ and $x_{j+1}$ as well as between $x_{j+1}$ and $B$ are uncrossed in $\Gamma_{j+1}$. 
The latter in conjunction with the above proposition yields that indeed, there is a face $f_{j+1}$ admitting invariant~\ref{i:f}.
To see this, note that $b_{j+1}$ cannot cross $g$, since it is not incident to $G$ (otherwise it would be parallel to $b$).
Therefore, no additional edges cross $g$ while extending the subdrawing from $\Gamma_j$ to $\Gamma_{j+1}$.
Invariant~\ref{i:f} can be combined with the fact that the arc of $b_{j-1}$ from $R$ to $x_j$ is uncrossed by invariant~\ref{i:un} to conclude that the triangular region $f_j\setminus f_{j+1}$ is  indeed empty and invariant~\ref{i:tri} is established.

This concludes the proof of the lemma in the case when $q_j = r_j$. The second case, $q_j=b_j$, is similar to the first one and can be found in Appendix~\ref{app:case2}.
\hfill$\square$
\end{proof}
Now that we concluded the proof of Lemma~\ref{lem:sequence_of_edges}, we have all the tools to prove Lemma~\ref{lem:adjacent_crossings_case2}.
\begin{lemma} \label{lem:adjacent_crossings_case2}
Let $\Gamma$ be a non-simple fan-planar drawing with the properties established by Corollary~\ref{cor:cor1}. If there is an edge $b=(G,R)$ in $\Gamma$ that crosses an edge $g$ at $x$ and $g$ is incident to $R$, then we can redraw an edge such that the total number of crossings in $\Gamma$ decreases, and the drawing remains fan-planar.
\end{lemma}

\begin{proof}
Let $b=(G,R)$ and $g=(B,R)$ be two adjacent edges which cross at $x$. Their common endpoint is not the special vertex of either of the edges by Corollary~\ref{cor:cor1}.
Thus, the special vertices of $b$ and $g$ must be $B$ and $G$, respectively.
We apply Lemma~\ref{lem:sequence_of_edges} on $b$ and $g$.
If $g$ can be redrawn using Lemma~\ref{lem:sequence_of_edges}, this concludes the proof of Lemma~\ref{lem:adjacent_crossings_case2}.
Assume that $g$ cannot be redrawn. Then we can determine a sequence of edges $r_0, b_1, r_2, \dots , r_k$ with the properties described in Lemma~\ref{lem:sequence_of_edges}. We now describe how the edge $b$ can be redrawn to eliminate the crossing $x$ while maintaining fan-planarity and decreasing the overall number of crossings.

Let the other endpoint of edge $r_k$ be $W$. 
By invariant~\ref{i:exit}, $r_k$ has a crossing with edge $b$. 
First assume this crossing occurs between $x_k$ and $W$, i.e., after $r_k$ enters the triangular region $f_{k-2}\setminus f_{k-1}$ at $x_k$. 
Since $b$ does not enter this region, $r_k$ has to leave it. It cannot cross $b_{k-1}$ again, nor can it cross $r_{k-2}$, because it is not incident to its special vertex $R$ (note that $W\neq R$ since otherwise $r_k$ would be parallel to $g$). 
Finally, it cannot cross $b_{k-3}$, because this is the part of $b_{k-3}$ that is uncrossed by invariant~\ref{i:un}. 
Hence, the crossing of $r_k$ and $b$ cannot lie between $x_k$ and $W$ and must instead lie between $B$ and $x_k$ along $r_k$.
In this case, we claim that edge $b$ can be redrawn. 
Redraw edge $b$ to follow $g$ from $R$ until $x$, and then follow its previous drawing from $x$ until $G$ while avoiding crossing $g$ at $x$, as illustrated in Figure~\ref{fig:reroute_b_fig}. 
We now prove that this redrawing does not introduce any new crossings on $b$.

\begin{figure}[ht]
  \centering
  \includegraphics[page=1,scale=1]
  {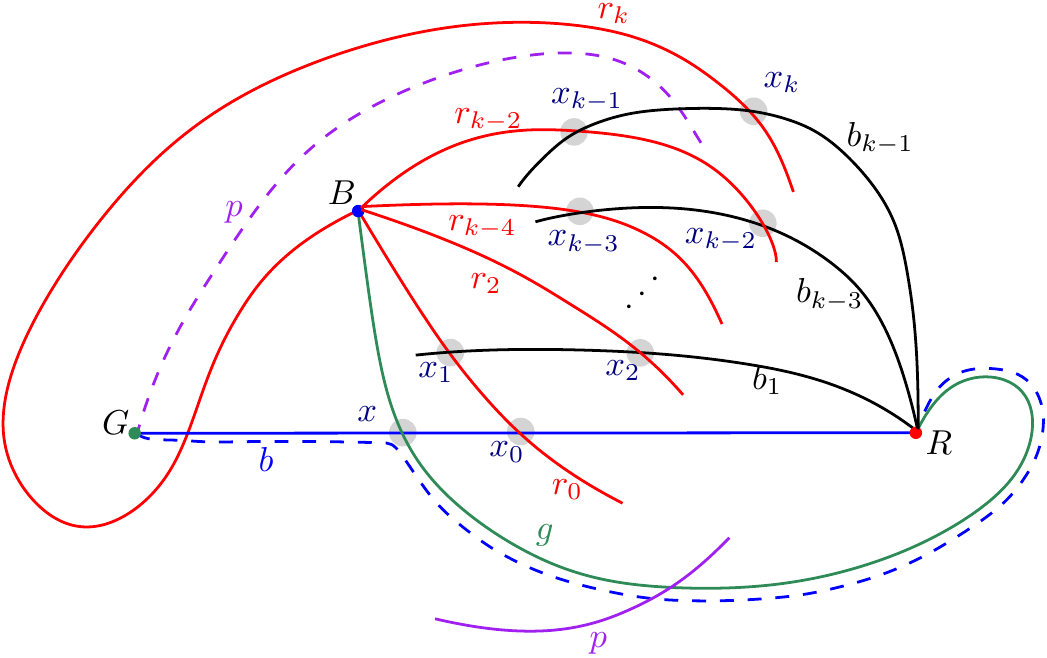}
  \caption{Redrawing of edge $b$.}
  \label{fig:reroute_b_fig}
\end{figure}

\begin{proposition}
Redrawing $b$ does not introduce any new crossings on $b$.
\end{proposition}
\begin{proof}
Assume a new crossing with an edge~$p$ is introduced on $b$ by the redrawing operation.
Since the redrawn part of~$b$ is parallel to a part of~$g$, the edge~$p$ crosses~$g$ as well. 
Consequently, edge $p$ is incident to $G$, the special vertex of $g$. 

Consider the closed curve $\delta$ formed by the arc of $r_k$ between $B$ and $x_k$, the arc of $r_{k-2}$ between $B$ and $x_{k-1}$, and the arc of $b_{k-1}$ between $x_{k-1}$ and $x_{k}$. 
The edge $b$ crosses $r_k$ exactly once, does not cross $r_{k-2}$ due to invariant~\ref{i:exit}, and also does not cross $b_{k-1}$ since the special vertex of $b_{k-1}$ is $B$ due to invariant~\ref{i:B} and $b$ is not incident to $B$. 
This implies that $b$ crosses $\delta$ exactly once and thus $R$ and $G$ have to lie on distinct sides of $\delta$. 
This is illustrated in Figure~\ref{fig:reroute_b_fig}. 
Edge $g$ does not cross any of the edges on the boundary of $\delta$ since $b$ is the only edge crossed by $g$ by Remark~\ref{rem:tri} except possibly for $r_0$, and even if $r_{k-2}=r_0$ the part of $r_{k-2}$ on $\delta$ is still uncrossed by invariant~\ref{i:un}, and therefore $g$ is contained in the same side of $\delta$ as its endpoint $R$.

The edge $p$ crosses $g$ and is incident to $G$, and thus must cross the curve $\delta$ since $g$ and $G$ lie on distinct sides of $\delta$. 
Edge $p$ cannot be incident to $R$ since then $p$ would be parallel to $b$. 
Since $R$ is the special vertex of $r_k$ and $r_{k-2}$ and $p$ is not incident to $R$, $p$ cannot cross the edges $r_k$ and $r_{k-2}$.
Hence, $p$ must cross edge $b_{k-1}$ to cross the curve $\delta$. 
Then the other endpoint of $p$ must be $B$, the special vertex of $b_{k-1}$.
However, the part of~$p$ connecting~$G$ with~$b_{k-1}$ is on the same side as the part of~$r_{k-2}$ between $B$ and $b_{k-1}$.
Consequently, the part of~$p$ connecting~$B$ to~$b_{k-1}$ and the part of~$r_{k-2}$ between $B$ and $b_{k-1}$ lie on distinct sides of~$b_{k-1}$, which contradicts the fan-planarity.
Overall, we have shown that~$p$ cannot cross~$\delta$ and, by extension, it cannot cross~$g$; a contradiction.
\hfill$\square$
\end{proof}

The only redrawn edge is $b$ and no new crossing is introduced on $b$, which ensures that fan-planarity is maintained. 
Additionally, we eliminate the crossing $x$, which decreases the total number of crossings in the drawing. 
\hfill$\square$
\end{proof}

We already described  in the beginning of Section~\ref{sec:main} how Lemmata~\ref{lem:adjacent_crossings_case1}--\ref{lem:adjacent_crossings_case2} can be combined to obtain a proof of Theorem~\ref{thm:main}; we formally summarize the proof in Appendix~\ref{app:main}.

%
%
\bibliographystyle{splncs04}
\bibliography{gd21_fanplanar}

\newpage
\appendix

\section{Proof of Lemma~\ref{lem:adjacent_crossings_case1}}
\label{app:lem1}
\begin{proof}
Every edge that crosses $b$ is incident to $B$. Traverse along the edge $b$ from $B$ to $R$, until an edge $g=(B, W)$ that crosses $b$ is encountered. Let this crossing be $x$.
We redraw the edge $g$ to follow the drawing of $b$ from $B$ until $x$ and then follow its previous drawing from $x$ to $W$ without crossing $b$ at $x$. The rerouting is illustrated in Figure~\ref{fig:adjacent_crossings_case1}.

\begin{figure}[ht]
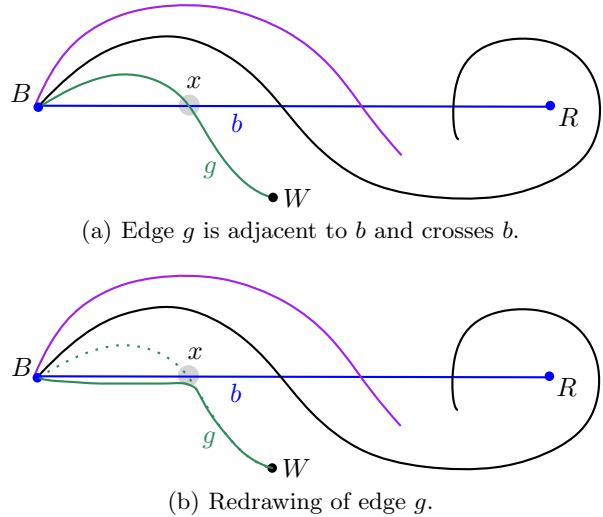

  \centering
  \subfigure[Edge $g$ is adjacent to $b$ and crosses $b$.]{
	 \label{fig:adjacent_crossings_case1_fig1}
	\includegraphics[scale=1]
	{adjacent_crossings_case1_fig1}}

  \subfigure[Redrawing of edge $g$.]{
	 \label{fig:adjacent_crossings_case1_fig2}
	\includegraphics[scale=1]
	{adjacent_crossings_case1_fig2}}

    \caption{Illustration of Lemma \ref{lem:adjacent_crossings_case1}. If $b$ is incident to its special vertex $B$, then all crossings on $b$ are adjacent crossings. We redraw the edge $g$ whose crossing $x$ with $b$ is closest to $B$ along $b$. Redrawing the part of $g$ between $x$ and $B$ along~$b$ cannot introduce any new crossings.}
  \label{fig:adjacent_crossings_case1}
\end{figure}

The part of edge $b$ between $B$ and $x$ has no crossings by definition of $x$.
Hence, the rerouting introduces no new crossings.
In particular, no new crossings are introduced on $g$ and, hence, fan-planarity is maintained.
Finally, since the crossing between $b$ and $g$ at $x$ is eliminated,  the total number of crossings decreases.
\hfill$\square$
\end{proof}

\section{Proofs of the propositions in Lemma~\ref{lem:multi_crossings}}
\label{app:prop12}
\begin{proof}[of Proposition~\ref{prop:z}]
If $p=g$, then we claim that $x=z$, i.e., $z$ is also the first crossing of $p(=g)$ with $b$. 
Assume otherwise that $z$ is the $i^{th}$ crossing of $p(=g)$ with $b$, where $i > 2$. 
Consider the closed curve $\delta$ formed by the arcs of $b$ and $g$ between $x$ and $y$. 
The arc of $p(=g)$ between $y$ and $W$ must cross the curve $\delta$ at $z$.
However, this implies that this arc of $p$ crosses itself or it crosses $b$ such that orienting $p$ towards $B$ and $b$ arbitrarily does not result in a consistent orientation of crossings. 
In both cases, we obtain a contradiction and, hence, $z$ is indeed the first crossing of $p(=g)$ with $b$ if $p=g$.

Now assume $p \neq g$ and assume that the crossing at $z$ is the $i^{th}$ crossing of $p$ with $b$, where $i \geq 2$.
Observe that in a fan-planar drawing, if an edge $e$ has three crossings $x_i,x_j,x_k$ with an edge $f$ such that $x_\ell$, where $\ell\in\{i,j,k\}$, is the $\ell^{th}$ crossing of $e$ with $f$ and $i<j<k$, then $x_i,x_j,x_k$ appear in this order along $f$.
Consequently, the first crossing of $p$ with $b$ has to be located between $z$ and $R$.
Otherwise, when traversing $b$ from $G$ to $R$, the first crossing of $p$ with $b$, the second crossing of $p$ with $b$, and the second crossing of $g$ with $b$ would be encountered in this order, which contradicts our choice of~$g$.
More specifically, the first crossing of $p$ with $b$, say $h$, must be between $y$ and $R$ since the arc of $b$ between $z$ and $y$ has no crossings by construction;
the situation is illustrated in Figure~\ref{fig:multi_crossings_fig1}. 
Consider the closed curve $\delta$ again. 
Since $p$ is incident to $B$ and $h$ is the first crossing of $p$ with $b$, the arc of $p$ starting from $B$ must cross $g$ between $x$ and $y$ to enter the region enclosed by $\delta$ before $p$ crosses $b$ at $h$. 
After $p$ crosses $b$ at $h$,
it has to cross $\delta$ again in order to cross $z$ such that the crossing orientation of $z$ and $h$ is consistent.
However, this second crossing with $\delta$ 
implies that the crossings of $p$ with $g$ or the crossings of $p$ with $b$ are not consistently oriented; a contradiction.

 Hence, in any case, $z$ is the first crossing of $p$ with $b$.
\hfill$\square$
\end{proof}

\begin{proof}[of Proposition~\ref{prop:multifan}]
To show that fan-planarity is maintained, we have to prove that the crossings introduced along the redrawn edge $g$ satisfy the conditions of fan-planarity. 
Observe that new crossings introduced along $g$ must be between $B$ and $z$ and the involved edges have to cross edge $p$ as well. 
Let an edge $r$ cross $p$ (and the redrawn version of $g$) at a point $c$ between $z$ and $B$.
If $p = g$, then fan-planarity is maintained since then the crossings of the redrawing of~$g$ are a subset of the crossings on its original drawing. 

Assume $p \neq g$. Let $\phi$ be the closed curve formed by the old drawing of $g$ between $B$ and $y$, the arc of $b$ between $y$ and $z$, and the arc of $p$ between $z$ and $B$. 
The edge $r$ must cross $\phi$ by definition since $c$ lies on $\phi$.
Since the second crossing of $g$ with $b$ is closer to $R$ than the first crossing and due to fan-planarity, $R$ and $G$ have to lie on distinct sides of $\phi$. At $c$, we split $r$ into two parts. We use $r^R$ to denote the part that enters the side of $\phi$ that contains $R$ -- the other part of $r$ is denoted by $r^G$.

The special vertex of $p$, say $P$, must be incident to edge $b$ since $b$ crosses $p$, and the edge $r$ must be incident to $P$ since $r$ also crosses $p$.
We distinguish two cases, namely $P=G$ and $P=R$.
We show that in both cases, $r$ crosses the part between $B$ and $y$ of the original drawing of $g$.

First, assume $P = G$. Since $P(=G)$ must be on a common side of $p$ for $r$ and $b$, the part of $r$ that is incident to $P(=G)$ has to be $r^R$.
 This implies that $r^R$ has to cross the curve $\phi$ by definition of $r^R$.
Let $s$ be the crossing of $r_R$ with curve $\phi$ that is closest to $c$ along $r_R$.
The crossing $s$ cannot lie on the arc of $p$ on $\phi$, since otherwise $r$ and $p$ would cross as in the configuration~SF2 which is forbidden.
Further, $s$ cannot lie on the arc of $b$ on $\phi$ since this part of $b$ is uncrossed by the definition of~$z$.
Hence, $s$ must lie on the old drawing of $g$ between $B$ and $y$, i.e., along the part of $\phi$ formed by the old drawing of $g$. This implies that $r$ crosses the part between $B$ and $y$ of the original drawing of $g$.

It remains to consider the case that $P=R$.
Since $P(=R)$ is on a common side of $p$ for $r$ and $b$, the part of $r$ that is incident to $P(=R)$ has to be $r^G$.
The arguments why $r$ crosses the part between $B$ and $y$ of the original drawing of $g$ are analogous to those used in the case $P=G$.

We have shown that $r$ crosses the part between $B$ and $y$ of the original drawing of $g$.
The corresponding crossing is eliminated when redrawing $g$, and a crossing between $r$ and $g$ is introduced after the redrawing.
Hence, even though the redrawn version of $g$ crosses $r$ between $z$ and $B$, the number of crossings does not increase.
Moreover, the orientation of the crossings between $r$ and $g$ is consistent with the orientation of the crossings in the redrawn version, i.e., fan-planarity is maintained.
\hfill$\square$
\end{proof}

\section{Proof of Corollary~\ref{cor:cor1}}
\label{app:cor1}

\begin{proof}[of Corollary~\ref{cor:cor1}]
The redrawing procedures guaranteed by Lemma~\ref{lem:multi_crossings} and~\ref{lem:adjacent_crossings_case1} decrease the number of crossings.
Hence, they can be exhaustively applied to~$\Gamma$ to obtain a drawing~$\Gamma'$ with no more crossings than~$\Gamma$ such that~$\Gamma'$ does not satisfy the precondition of  Lemma~\ref{lem:adjacent_crossings_case1} or~\ref{lem:multi_crossings}.
In particular, if an edge is incident to its special vertex, all edges crossing it must be adjacent to it. If there is such an edge, Lemma \ref{lem:adjacent_crossings_case1} is applicable. If there is no such edge, we may choose a new special vertex for that edge, which is not incident to it.
Hence, $\Gamma'$ has the desired properties.
\hfill$\square$
\end{proof}

\section{Proof of the proposition in the base case of Lemma~\ref{lem:sequence_of_edges}}
\label{app:base}

\begin{proof}[of Proposition~\ref{prop:base}]
Assume otherwise and consider the closed curve~$\gamma$ formed by the parts of~$g$ and~$b$ that lead from~$R$ to~$x$.
Both~$G$ and~$B$ are on the same side of~$\gamma$ since there are no multiple crossings.
Since~$r_0$ intersects~$g$, it follows that~$r_0$ is incident to the special vertex~$G$ of~$g$.
Let the arc of $r_{0}$ between $B$ and $x_0$ be denoted by $r_0^B$ and the arc between $G$ and $x_0$ be denoted by $r_0^G$.
Since $B$ must lie on the same side of $b$ with respect to the two crossings $x$ and $x_0$, the arc $r_0^G$ must be the arc on the side of $\gamma$ which does not contain $B$.
However, for $r_0^G$ to be incident to $G$, $r_0^G$ must cross~$\gamma$, thereby crossing~$b$ or~$g$ a second time (recall that $r_0^B$ intersects $g$ by assumption), arriving at a contradiction.
\hfill$\square$
\end{proof}

\section{Proof of the proposition in the inductive step of Lemma~\ref{lem:sequence_of_edges}: Case 1}
\label{app:case1}

\begin{figure}[ht]
  \centering
  \includegraphics[scale=1]
  {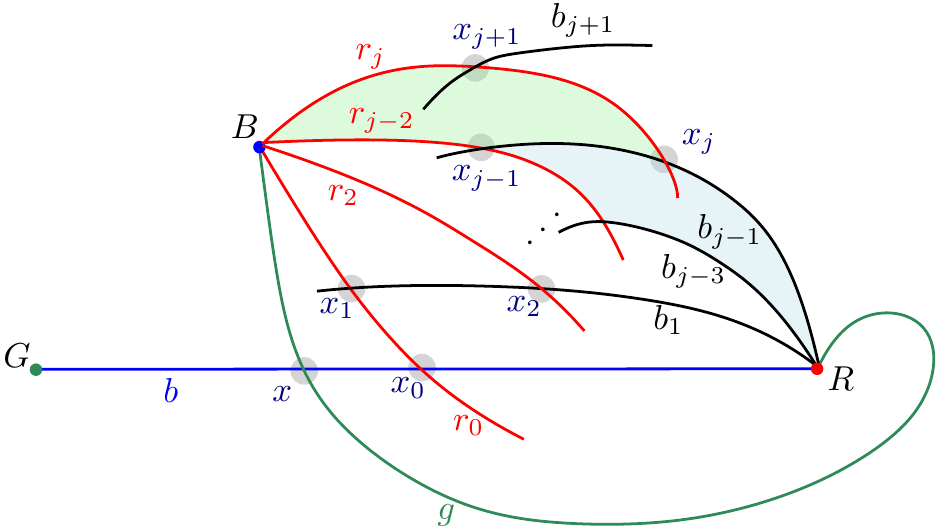}
  \caption{Determining $b_{j+1}$.}
  \label{fig:adjacent_crossings_case2_fig4}
\end{figure}

\begin{proof}[of Proposition~\ref{prop:R}]
Assume $b_{j+1}$ is not incident to $R$. Since $r_j$ crosses $b_{j-1}$ and $b_{j+1}$, they must have a common endpoint, which is not $R$ by assumption. Let the other endpoint of $b_{j-1}$ be $Z$, which implies $b_{j+1}$ is also incident to $Z$. Let the arc of $b_{j+1}$ which exits the face $f_j$ and enters region $f_{j-1} \setminus f_j$ (which is shaded in green in Figure~\ref{fig:adjacent_crossings_case2_fig4}) be denoted by $b_{j+1}^{i}$. To maintain fan-planarity, the vertex $Z$ must be on the same side of $r_j$ along both edges $b_{j-1}$ and $b_{j+1}$, which implies that the arc $b_{j+1}^{i}$ must end at $Z$. Endvertex $Z$ is not in the region $f_{j-1} \setminus f_j$ by invariant~\ref{i:tri}, hence the arc $b_{j+1}^{i}$ must exit the region $f_{j-1} \setminus f_j$. To exit the region $f_{j-1} \setminus f_j$, arc $b_{j+1}^{i}$ must cross $r_j$, $b_{j-1}$ or $r_{j-2}$. Arc $b_{j+1}^{i}$ cannot cross $r_j$ again, and also cannot cross $r_{j-2}$ since the arc of $r_{j-2}$ in this region has no crossings by invariant~\ref{i:un}. Hence, $b_{j+1}^{i}$ must cross $b_{j-1}$ to exit the region $f_{j-1} \setminus f_j$, and must enter the region $f_{j-2} \setminus f_{j-1}$ (which is shaded in blue in Figure~\ref{fig:adjacent_crossings_case2_fig4}). Again, endpoint $Z$ is not in the region $f_{j-2} \setminus f_{j-1}$ by invariant~\ref{i:tri}, and hence the arc $b_{j+1}^{i}$ must exit the region $f_{j-2} \setminus f_{j-1}$ by crossing $b_{j-1}$, $b_{j-3}$ or $r_{j-2}$. However, $b_{j+1}^{i}$ cannot cross $b_{j-1}$ again, cannot cross $b_{j-3}$ since the arc of $b_{j-3}$\footnote{If $j=2$ this denotes $b_{-1}=b$ again} in this region is uncrossed by invariant~\ref{i:un}, and also cannot cross $r_{j-2}$ since $b_{j+1}$ must be incident to $R$ to cross $r_{j-2}$ since $R$ is the special vertex of $r_{j-2}$. Hence, $b_{j+1}^{i}$ cannot exit the region $f_{j-2} \setminus f_{j-1}$ and thus cannot be incident to $Z$. We arrive at a contradiction, which implies $b_{j+1}$ is incident to $R$.
\hfill$\square$
\end{proof}

\begin{figure}[ht]
  \centering
  \includegraphics[scale=1]
  {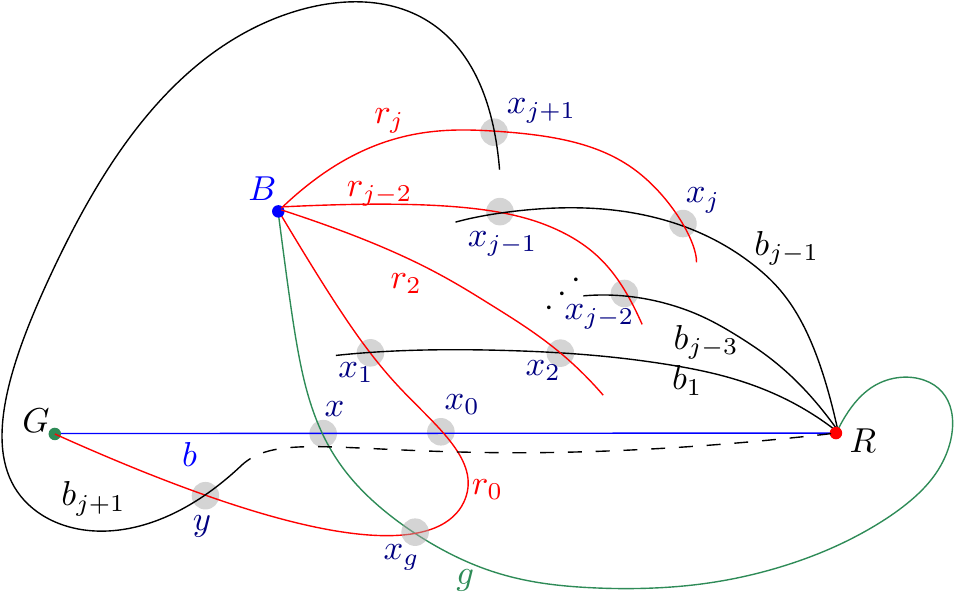}
  \caption{Illustrates the case that $b_{j+1}$ crosses $r_0$ directly after $x_{j+1}$.}
  \label{fig:adjacent_crossings_case2_fig5}
\end{figure}
\begin{proof}[of Proposition~\ref{prop:r0}]
Edge $b_{j+1}$ is incident to $R$, let the other endpoint be $S$. Let the arc of $b_{j+1}$ which exits the face $f_j$ at $x_{j+1}$ and enters region $f_{j-1} \setminus f_j$ be denoted by $b_{j+1}^{i}$. To exit the region $f_{j-1} \setminus f_j$, arc $b_{j+1}^{i}$ would have to cross $r_j$, $b_{j-1}$ or $r_{j-2}$. Arc $b_{j+1}^{i}$ cannot cross $r_j$ again, and also cannot cross $r_{j-2}$ since the arc of $r_{j-2}$ in this region is uncrossed by invariant~\ref{i:un}. For arc $b_{j+1}^{i}$ to cross $b_{j-1}$, edge $b_{j+1}$ must be incident to $B$, the special vertex of $b_{j-1}$. However, $b_{j+1}$ would be parallel to $g$ in that case and the graph has no parallel edges, so this would only be possible if $g=b_{j+1}$, but $g$ does not cross $b_{j-1}$ either, which implies arc $b_{j+1}^{i}$ cannot cross $b_{j-1}$. Thus, the arc $b_{j+1}^{i}$ cannot exit the region $f_{j-1} \setminus f_j$. This shows that the region $f_{j-1}\setminus f_j$ contains an endpoint of $b_{j+1}$. Since this region was empty in $\Gamma_j$ by invariant~\ref{i:tri}, $b_{j+1}$ is distinct from all the previous edges.

Assume edge $b_{j+1}$ crosses edge $r_0$. If the arc $b_{j+1}^{i}$ crosses $r_0$, the crossing must be in region $f_{j-1} \setminus f_j$.  This region was empty in $\Gamma_j$ though, so $r_0$ does not enter it. This also shows $b_{j+1}$ is not incident to $B$ or $G$. 

Therefore $r_0$ has to cross $b_{j+1}$ along the arc of $b_{j+1}$ between $R$ and $x_{j+1}$ at a crossing~$y$. 
We know $b_{j+1}$ is not incident to $G$ nor $B$, so it doesn't cross $g$ nor $b_{j-1}$. 
Therefore this arc lies entirely in $f_j$, since none of the arcs $b_{j-1}$, $r_j$ and $g$ can be crossed by it. 
$r_0$ leaves $f_1\supset f_j$ when crossing $b_1$ at $x_1$. $r_0$ has to cross $\alpha_1$ again in order to re-enter $f_1$. 
To cross $\alpha_1$ again, $r_0$ has to cross $g$, since it cannot cross itself nor $b_1$ again. 
However, if $r_0$ crosses $g$ at a crossing $x_g$, the other end point of $r_0$ must be $G$. The edge $r_0$ thus will enter $f_1$ again after crossing $g$ at $x_g$, crosses $b_{j+1}$ at $y$ and then ends to $G$. 
Consider the arc of $b_{j+1}$ that starts at $y$ and ends at $R$.
After crossing $r_0$ at $y$, this arc enters the triangle $x,x_g,G$ bounded by $b,g$ and $r_0$, see Figure~\ref{fig:adjacent_crossings_case2_fig5} for an illustration. 
However, it is impossible for this arc to exit the triangle since none of the edges bounding the triangle can be crossed (again) by $b_{j+1}$ as shown above, and thus the arc of $b_{j+1}$ cannot be incident to $R$, a contradiction.
\hfill$\square$
\end{proof}

\begin{proof}[of Proposition~\ref{prop:uncross}]
Towards a contradiction assume that the arc of $b_{j+1}$ between $R$ and $x_{j+1}$ is crossed in $\Gamma_{j+1}$. Then it must be crossed by $b$, $g$, or $q_{i}$, where $0 \leq i < k,$ $ q\in\{r,b\}$. 

Let $y$ be the crossing along the arc of $b_{j+1}$ between $x_{j+1}$ and $R$ that is closest to $x_{j+1}$. $b_{j+1}$ is neither incident to $G$ nor $B$, since it otherwise would be parallel to $b$ or $g$ respectively. Therefore $y$ cannot lie on any edge with these special vertices. This excludes the black edges, $b$ and $g$. By Proposition~\ref{prop:r0}, $b_{j+1}$ also does not cross $r_0$.

Further $y$ cannot lie on the arc of $r_j$ between $B$ and $x_{j+1}$ since this arc is uncrossed by invariant~\ref{i:un}. Thus, $y$ cannot lie on the curve $\alpha_j$, and must lie in the interior of $f_j$.

Since $b_{j+1}$ crosses $r_j$ at $x_{j+1}$ already, $y$ cannot lie on $r_j$.

Hence, $y$ must lie on a red edge $r_i$ where $1 \leq i < j$. Let $\delta$ be the closed curve formed by the arc of $b$ between $B$ and $x_0$, the arc of $r_0$ between $R$ and $x_0$, the arc of $r_j$ between $B$ and $x_j$ and the arc of $b_{j-1}$ between $R$ and $x_j$.  Invariant~ \ref{i:tri} applied to $i, i+1,i+2$ shows that the entirety of edge $r_i$ is outside $f_{i+2}\supset f_j$. Invariant~\ref{i:tri} applied to $i,i-1,i-2$ shows that it is completely contained in $f_{i-2}\subset f_0$. Now let $\delta=\alpha_0\triangle \alpha_j$, 
then edge $r_{i}$ must cross $\delta$ since $y$ lies in $f_j\subset f_0$, on the other side of the curve. To cross $\delta$, $r_i$ must cross one of the arcs of the curve. $r_i$ cannot cross $b$ or $b_{j-1}$ since these arcs are uncrossed by invariant~\ref{i:un}, and cannot cross $r_0$ or $r_j$ since it is not incident to $R$, since invariant~\ref{i:R} would imply it is an edge parallel to $g$. Thus, $y$ does not lie on a red edge.

This proves the proposition that the arc of $b_{j+1}$ between $x_{j+1}$ and $R$ is uncrossed in the drawing $\Gamma_{j+1}$.
\hfill$\square$
\end{proof}
\section{Proof of the inductive step of Lemma~\ref{lem:sequence_of_edges}: Case 2}
\label{app:case2}
\paragraph{Case 2:} $q_j = b_j$. Note that in this case, there is nothing to prove for invariant~\ref{i:R}. 

If the edge $b_j$ has no crossings between $R$ and $x_j$, then we could redraw $g$ along this part of $b_j$ and along the arc of $r_{j-1}$ from $x_j$ to $B$. $g$ would then be uncrossed by invariant~\ref{i:un} and the lemma is proved.

 We determine the edge $r_{j+1}$ as follows: traverse along $b_j$ from $R$ towards $x_j$ until we encounter the first edge which crosses $b_j$. Let this edge be $r_{j+1}$. We prove that $r_{j+1}$ satisfies the invariants.

Invariant~\ref{i:un} is true by definition. We now prove invariant~\ref{i:B} for $b_j$ and $r_{j+1}$.

\begin{proposition}
Edge $r_{j+1}$ is incident to $B$.
\end{proposition}
\begin{proof}
Assume $r_{j+1}$ is not incident to $B$. Since $b_j$ crosses $r_{j-1}$ and $r_{j+1}$, both the edges $r_{j-1}$ and $r_{j+1}$ must have a common endpoint, which is not $B$ by assumption. Let the other endpoint of $r_{j-1}$ be $Z$, which implies $r_{j+1}$ is also incident to $z_{j-1}$.
Let the arc of $r_{j+1}$ which exits the face $f_j$ and enters region $f_{j-1} \setminus f_j$ (which is shaded in blue in Figure~\ref{fig:adjacent_crossings_case2_fig2}) be denoted by $r_{j+1}^{i}$. To maintain fan-planarity, the vertex $Z$ must be on the same side of $b_j$ along both edges $r_{j-1}$ and $r_{j+1}$, which implies that the arc $r_{j+1}^{i}$ must end at $Z$. Endvertex $Z$ is not in the region $f_{j-1} \setminus f_j$ by invariant~\ref{i:tri}, hence the arc $r_{j+1}^{i}$ must exit the region $f_{j-1} \setminus f_j$. To exit the region $f_{j-1} \setminus f_j$, arc $r_{j+1}^{i}$ must cross $b_j$, $r_{j-1}$ or $b_{j-2}$. Arc $r_{j+1}^{i}$ cannot cross $b_j$ again, and also cannot cross $b_{j-2}$ since the arc of $b_{j-2}$ in this region is uncrossed by invariant~\ref{i:un}. Hence, $r_{j+1}^{i}$ must cross $r_{j-1}$ to exit the region $f_{j-1} \setminus f_j$. If $j>1$, it enters the region $f_{j-2} \setminus f_{j-1}$ (which is shaded in green in Figure~\ref{fig:adjacent_crossings_case2_fig2}). 
Again, endpoint $Z$ is not in the region $f_{j-2} \setminus f_{j-1}$ by invariant~\ref{i:tri}, and hence the arc $r_{j+1}^{i}$ must exit the region $f_{j-2} \setminus f_{j-1}$ by crossing $r_{j-1}$, $r_{j-3}$ or $b_{j-2}$. However, $r_{j+1}^{i}$ cannot cross $r_{j-1}$ again, cannot cross $r_{j-3}$ since the arc of $r_{j-3}$ in this region is uncrossed by invariant~\ref{i:un}, and also cannot cross $b_{j-2}$ since $r_{j+1}$ must be incident to $B$ to cross $b_{j-2}$ since $B$ is the special vertex of $b_{j-2}$. Hence, $r_{j+1}^{i}$ cannot exit the region $f_{j-2} \setminus f_{j-1}$ and thus cannot be incident to~$Z$. 
 
In the special case $j=1$, if $r_2$ crosses $b_1$ and then $r_0$, then the two edges $r_0$ and $b_1$ must have a vertex in common and since neither edge is $g$, it cannot be $R$ or $B$. However, orienting $r_2$ from $B$ to $Z$, $R$ is on the left at $x_2$ and $B$ is on the right. Therefore the other endpoint of $r_0$ and $b_1$, $Z$, will lie on different sides of $r_2$ with respect to $r_0$ and $b_1$, and hence cannot be a special vertex of $r_2$.
In both cases, we arrive at a contradiction, so $r_{j+1}$ is incident to $B$.
\hfill$\square$
\end{proof}

Since $b_j$ crosses both edges $r_{j+1}$ and $r_{j-1}$, which are both incident to $B$, the special vertex of $b_{j}$ is $B$, which proves invariant~\ref{i:B}. 

\begin{figure}[ht]
  \centering
  \includegraphics[scale=1]
  {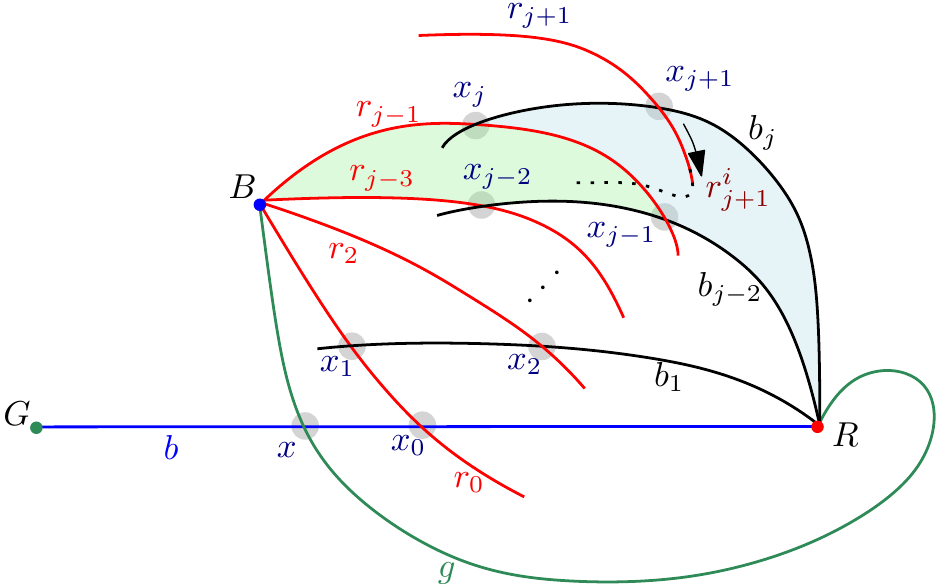}
  \caption{Determining edge $r_{j+1}$.}
  \label{fig:adjacent_crossings_case2_fig2}
\end{figure}

Next, we prove that $r_{j+1}$ is distinct from the previous edges.

\begin{proposition}
 $r_{j+1}$ is distinct from the edges of $\Gamma_j$.
\end{proposition}
\begin{proof}
Similarly to the last proof, let the arc of $r_{j+1}$ which exits the face $f_j$ and enters region $f_{j-1} \setminus f_j$ be denoted by $r_{j+1}^{i}$. We claim that this arc cannot leave this region anymore. To do so, it would either have to cross $b_j$ again, or cross $b_{j-2}$ in the part uncrossed by invariant~\ref{i:un}, or cross $r_{j-1}$. In that case, it is incident to $R$ though, and would therefore be $g$, since there are no parallel edges. It is not $g$ though, because $b_j$ does not cross $g$. Therefore, the other endpoint of $r_{j+1}$ lies inside $f_{j-1} \setminus f_j$, an empty region in $\Gamma_j$, establishing the proposition.
\hfill$\square$
\end{proof}

Now, to prove invariants~\ref{i:exit},~\ref{i:f}, and~\ref{i:tri}, we prove that $r_{j+1}$ is uncrossed between $x_{j+1}$ and $B$ in the drawing $\Gamma_{j+1}$, unless $j+1 = k$. 

\begin{proposition} If $j+1=k$, then all invariants hold and the arc of $r_{j+1}$ between $x_{j+1}$ and $B$ crosses $b$.
Otherwise, this arc is uncrossed in the drawing~$\Gamma_{j+1}$.
\end{proposition}
\begin{proof}
Towards a contradiction assume that the arc of $r_{j+1}$ between $x_{j+1}$ and $B$ is crossed in $\Gamma_{j+1}$. Then it must be crossed by $b$, $g$ or $q_{i}$, where $i\in\{0,\ldots,k\},$ $q\in \{r,b\}$.

Let $y$ be the crossing along the arc of $r_{j+1}$ between $x_{j+1}$ and $B$ that is closest to $x_{j+1}$. If $y$ lies on the edge $b$, then invariant~\ref{i:exit} is satisfied, $i = k$, and the sequence is maximal. Invariants~\ref{i:f} and \ref{i:tri} are void in this case, so they hold trivially. Invariants~\ref{i:B} and \ref{i:un} have already been proved. Finally, invariant~\ref{i:R} is true for $r_{j+1}$, because it crosses $b_j$ and $b$, both of which are incident to $R$, so $R$ has to be its special vertex.

Assume $y$ lies on the closed curve $\alpha_j$. Recall that the curve $\alpha_j$ is formed by the edge $g$, the arc of $b_j$ between $R$ and $x_j$ and the arc of $r_{j-1}$ between $B$ and $x_j$. The crossing $y$ must lie on $g$ since $r_{j+1}$ cannot cross the arc of $b_j$ again and the arc of $r_{j-1}$ is uncrossed due to invariant~\ref{i:un}.
We know that the special vertex of $g$ is $G$, otherwise the crossing between $b$ and $g$ would have been eliminated by Corollary~\ref{cor:cor1}. This implies that $r_{j+1}$ must be incident to $G$ since the edge $g$ is crossed by $r_{j+1}$ if $y$ lies on $g$. Thus, the other endpoint of $r_{j+1}$, which is the endpoint of arc $r_{j+1}^i$, must be $G$. For arc $r_{j+1}^i$ to end in $G$, it must cross the curve $\alpha_j$. The arc $r_{j+1}^i$ cannot cross the arc of $r_{j-1}$ along $\alpha_j$ since the arc of $r_{j-1}$ along $\alpha_j$ is uncrossed by invariant~\ref{i:un}, and it cannot cross the arc of $b_j$ nor $g$ again. Hence, the crossing $y$ is not on the curve $\alpha_j$, which implies that $y$ lies in the interior of $f_j$. 
However, the interior of $f_j$ contains only $b$ and possibly $r_0$ due to invariant~\ref{i:f}.

 Edge $r_{j+1}$ cannot be incident to $R$, since otherwise $r_{j+1}$ would be an edge that is parallel to $g$. Therefore $r_{j+1}$ cannot cross $r_0$ whose special vertex is $R$. 

\begin{figure}[ht]
  \centering
  \includegraphics[scale=1]
  {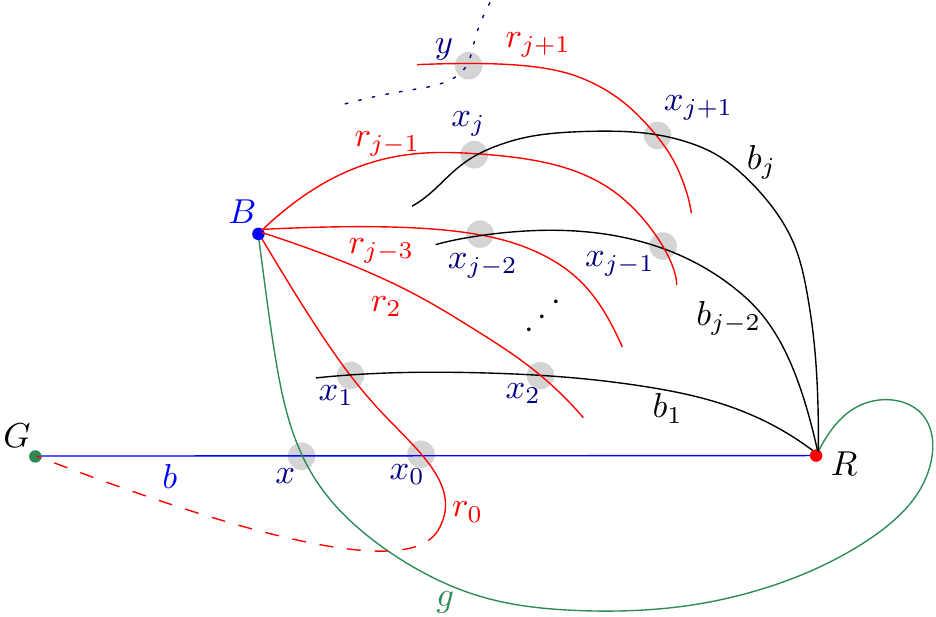}
  \caption{Illustration of the case that $r_{j+1}$ is crossed between $B$ and $x_{j+1}$.}
  \label{fig:adjacent_crossings_case2_fig3}
\end{figure}

This proves the proposition that the arc of $r_{j+1}$ between $x_{j+1}$ and $B$ is uncrossed in the drawing $\Gamma_{j+1}$, unless $j+1 = k$.
\hfill$\square$
\end{proof}

Assume $r_{j+1} \neq q_k$. Then the arc of $r_{j+1}$ between $x_{j+1}$ and $B$ is uncrossed in the drawing $\Gamma_{j+1}$. 
Further, the arc of $b_j$ between $x_j$ and $R$ was uncrossed in $\Gamma_j$ by invariant~\ref{i:tri}. Since $r_{j+1}$ is the only new edge introduced for $\Gamma_{j+1}$, the arcs of $b_j$ between $x_j$ and $x_{j+1}$ as well as between $x_{j+1}$ and $R$ are uncrossed. 
The latter can be used in conjunction with the above proposition to conclude, that indeed, there is a face $f_{j+1}$ admitting invariant~\ref{i:f}.  To see this, note that $r_{j+1}$ cannot cross $g$, since it is not incident to $G$. Therefore, no additional edges cross $g$ when extending the drawing from $\Gamma_j$ to $\Gamma_{j+1}$.
Invariant~\ref{i:f} can be combined with the fact that the arc of $r_{j-1}$ from $B$ to $x_j$ is uncrossed by invariant~\ref{i:un} to conclude that the triangular region $f_j\setminus f_{j+1}$ is empty and invariant~\ref{i:tri} is established. 
If the arc of $r_{j+1}$ between $x_{j+1}$ and $B$ is uncrossed, then it does not intersect $b$. On the other hand, we have already proved that the arc of $r_{j+1}$ that enters $f_j\setminus f_{j+1}$ at $x_{j+1}$ cannot leave this region, and $f_j\setminus f_{j+1}$ does not contain any part of $b$ by invariant~\ref{i:tri}. This finally proves invariant~\ref{i:exit}, and concludes the proof of the lemma in the case where $q_j = b_j$.

\section{Proof of Theorem~\ref{thm:main}}
\label{app:main}

\begin{proof}[of Theorem~\ref{thm:main}]
Let $\Gamma$ be a non-simple fan-planar drawing. 
By Corollary~\ref{cor:cor1} we can find a fan-planar redrawing $\Gamma'$ of $\Gamma$ such that
\begin{itemize}
\item no two edges cross more than once in $\Gamma'$,
\item if two adjacent edges cross in $\Gamma'$, then their common endpoint is not the special vertex of either of the two edges;, and
\item $\Gamma'$ does not have more crossings than $\Gamma$.
\end{itemize}
If $\Gamma'$ still contains adjacent crossings, we can apply Lemma~\ref{lem:adjacent_crossings_case2} to obtain a fan-planar redrawing of~$\Gamma'$ with fewer crossings.
Since the number of crossings is finite, we can iterate this procedure to eventually obtain a simple fan-planar drawing.
\hfill$\square$
\end{proof}

\end{document}